\newif\ifarxiv
\arxivtrue

\ifarxiv
	\documentclass[12pt]{article}
	\usepackage[margin=1in]{geometry}
	\usepackage[slantedGreek]{mathpazo}
	\usepackage[pdftex]{graphicx}
	\usepackage{amsfonts}
	\usepackage{amssymb}
	\usepackage{amsthm}
	\usepackage{graphicx,float}
	\usepackage{amsmath}%
	\usepackage{hyperref}
	\usepackage[numbers,sort&compress]{natbib}
\else
\fi
\usepackage{color,soul,marginnote}
\usepackage{booktabs}
\usepackage{tikz}
\usetikzlibrary{arrows.meta,positioning,calc,decorations.pathmorphing}

\DeclareMathOperator{\tr}{tr}

\newtheorem{theorem}{Theorem}
\newtheorem{proposition}[theorem]{Proposition}

\newtheorem{definition}[theorem]{Definition}

\ifarxiv
\title{Bell's Inequality, Causal Bounds, and Quantum Bayesian Computation:\\ A Unified Framework}
\author{
	\makebox[.3\linewidth]{Nicholas G.\ Polson}\\\textit{Booth School of Business}\\\textit{University of Chicago}\\\and
	\makebox[.3\linewidth]{Vadim Sokolov}\\\textit{Department of Systems Engineering}\\\textit{and Operations Research}\\\textit{George Mason University}\\\and
	\makebox[.3\linewidth]{Daniel Zantedeschi}\\\textit{College of Business}\\\textit{University of South Florida}
}
\date{First Draft: March 2026\\This Draft: \today}
\else
\fi

\graphicspath{{./fig/}}

\begin{document}
\ifarxiv
\maketitle
\begin{abstract}
Bell inequalities characterize the boundary of the local-realist correlation polytope---the set of joint probability distributions achievable by classical hidden-variable models. Quantum mechanics exceeds this boundary through non-commutativity, reaching the Tsirelson bound $2\sqrt{2}$ for CHSH. We show that this polytope structure is not specific to quantum foundations: it appears identically in the causal inference literature, where the instrumental inequality, the Balke--Pearl linear programming bounds, and the Tian--Pearl probabilities of causation all arise as facets of the same marginal compatibility polytope. Fine's theorem---that CHSH inequalities hold if and only if a joint distribution exists---is precisely the pivot: the instrumental variable model in causal inference is structurally equivalent to the Bell local hidden-variable model, with the instrument playing the role of the measurement setting and the latent confounder playing the role of the hidden variable $\lambda$. We develop this correspondence in detail, extending it to algorithmic (Kolmogorov complexity) and entropic formulations of Bell inequalities, the NPA semidefinite programming hierarchy, and the MIP$^*$=RE undecidability result. We further show that the Born-rule / Bayes-rule duality underlying quantum Bayesian computation exploits the same non-commutativity that enables Bell violation, providing polynomial speedups for posterior inference. The framework yields a concrete dictionary between quantum information theory, causal econometrics, and Bayesian computation, and suggests new directions including NPA-based quantum causal inference algorithms and quantum architectures for function approximation.
\end{abstract}
\else
\fi


\section{Introduction}\label{sec:intro}

Bell's 1964 theorem established that the correlations between measurements on entangled quantum particles cannot be reproduced by any joint probability distribution over local hidden variables \cite{bell1964einstein}. The CHSH generalization \cite{clauser1969proposed} made this testable: the correlation polytope of local-realist models has facets $|S| \leq 2$, while quantum mechanics reaches $|S| = 2\sqrt{2}$ at the Tsirelson bound \cite{tsirelson1980quantum}. The NPA hierarchy \cite{navascues2007bounding, navascues2008convergent} provides a converging sequence of semidefinite programs (SDPs) that outer-approximate the quantum correlation set, and the MIP$^*$=RE theorem \cite{ji2021mip} shows that membership in this set is undecidable in general. These results---polytope characterization, SDP relaxation, undecidability---constitute the modern theory of quantum correlations (see \cite{brunner2014bell} for a comprehensive review).

Viewed abstractly, the Bell scenario poses a \emph{marginal compatibility problem}: given observed two-party correlations, does a joint distribution over all measurement outcomes exist? This question is not specific to quantum physics. \citet{boole1854investigation} studied it as a problem in logic. \citet{frechet1951tableaux} studied it as a coupling problem. And in causal inference---a field largely unknown to the quantum information community---the identical mathematical structure has been developed independently over three decades. \citet{robins2015bell} proved Bell's inequality directly from the theory of causal interactions; \citet{gill2023bells} showed that the CHSH inequalities follow as an exercise in the modern DAG-based theory of statistical causality.

In econometrics and biostatistics, the fundamental obstacle is that potential outcomes under different treatments are never jointly observed for the same individual. \citet{manski1990nonparametric} showed that without structural assumptions, treatment effects are only \emph{partially identified}: they lie in an interval determined by Fr\'{e}chet bounds on the unobserved joint. \citet{balke1997bounds} derived tight bounds on causal effects via linear programming over a \emph{response-function polytope}---the convex hull of deterministic response types, which is precisely the local-realist polytope of Bell's scenario restricted to the instrumental variable (IV) DAG. \citet{pearl1995testability} derived an \emph{instrumental inequality} that is the causal analog of the CHSH inequality: a testable constraint on observed distributions that must hold if the IV structural assumptions are valid. \citet{tian2000probabilities} bounded counterfactual probabilities of necessity and sufficiency.

The correspondence is exact. Fine's theorem \cite{fine1982hidden}---that CHSH inequalities hold if and only if a joint distribution exists---is the pivot. The IV model maps onto the Bell local hidden-variable model: the instrument $Z$ is the measurement setting, the treatment $X$ is Alice's outcome, the observed outcome $Y$ is Bob's outcome, and the latent confounder $U$ is the hidden variable $\lambda$. The inflation technique of \citet{wolfe2019inflation} subsumes both Bell inequalities and instrumental inequalities as special cases of compatibility constraints on causal DAGs with latent variables.

The unification yields transferable tools in both directions. The NPA SDP hierarchy, developed for bounding quantum correlations, can be applied to compute bounds on causal effects in DAGs with latent confounders. Conversely, the response-type linear programming machinery of causal inference provides an explicit vertex enumeration of the local-realist polytope that complements the facet description given by Bell inequalities.

A second thread connects this polytope geometry to quantum computation. The Born rule in quantum mechanics and Bayes' rule in statistics are both conditioning operations---one on density matrices, the other on probability distributions. When the density matrix is diagonal, the two coincide \cite{polson2023quantum}. When it is not, quantum coherence enables correlations that no classical joint distribution can reproduce---the same non-commutativity that produces Bell violation provides computational speedup for posterior inference. The Kolmogorov Superposition Theorem provides the classical counterpart, and K-GAM networks \cite{polson2025kgam} implement this decomposition with the horseshoe prior selecting the minimal architecture.

We make the following contributions:
\begin{enumerate}
\item A \emph{unified polytope framework} connecting Bell inequalities, Fr\'{e}chet bounds, Pearl's instrumental inequality, and Manski's partial identification as facets of the same class of marginal-compatibility problems.
\item \emph{Kolmogorov complexity bounds} as the algorithmic version of Bell inequalities, characterizing the minimum description length of joint quantum information and connecting to the Kakeya conjecture via the Lutz conditional-dimension principle.
\item \emph{Quantum Bell inequalities}---the Tsirelson hierarchy, the NPA semidefinite program, the MIP$^*$=RE resolution---as the operator-algebraic extension of the classical polytope, where non-commutativity enables correlations strictly beyond the Fr\'{e}chet-feasible region.
\item \emph{Quantum Bayesian Computation} (QBC) as the computational realization of the Born-rule / Bayes-rule duality, showing that the same non-commutativity enabling Bell violation provides polynomial-to-exponential speedup in posterior sampling (quadratic for MCMC, up to exponential for linear algebra conditional on favorable condition numbers).
\item \emph{K-GAM and the horseshoe prior} as the classical approximation to quantum function factorization, where Kolmogorov's Superposition Theorem provides the classical analog of quantum superposition.
\end{enumerate}

Section~\ref{sec:polytope} develops the classical polytope framework. Section~\ref{sec:causal} maps this framework onto causal inference. Section~\ref{sec:kolmogorov} reformulates Bell inequalities in the language of Kolmogorov complexity and Shannon entropy. Section~\ref{sec:quantum} extends the polytope to the quantum setting via the Tsirelson bound and the NPA hierarchy. Section~\ref{sec:qbc} presents quantum Bayesian computation and its classical counterpart in Generative Bayesian Computation (GBC) and K-GAM networks. Section~\ref{sec:unified} collects these threads into a unified polytope dictionary. Proofs are in the Appendix.

\section{The Classical Feasibility Polytope}\label{sec:polytope}

We begin with the mathematical core: the set of joint probability distributions compatible with observed marginals. This object---a convex polytope in the space of joint distributions---underlies every bound in the paper. The key intuition is simple: observing marginals constrains but does not determine the joint. The feasible set of joints is a convex body whose geometry encodes what can and cannot be inferred from marginal data alone. Bounds on functionals of the joint (treatment effects, correlations, counterfactual probabilities) are optimization problems over this polytope.

\subsection{Fr\'{e}chet--Hoeffding Bounds}

Let $F_0, F_1$ be univariate distribution functions. The Fr\'{e}chet problem asks: what is the set of bivariate distributions $F$ with marginals $F_0$ and $F_1$? In the causal inference application (Section~\ref{sec:causal}), $F_0$ and $F_1$ are the marginal distributions of potential outcomes $Y(0)$ and $Y(1)$, but the result is purely probabilistic.

\begin{theorem}[Fr\'{e}chet, 1951]\label{thm:frechet}
The set of feasible joint distributions $\mathcal{F}(F_0, F_1)$ is the convex set bounded by
\[
W(u,v) = \max(u + v - 1, 0) \leq F\bigl(F_0^{-1}(u), F_1^{-1}(v)\bigr) \leq \min(u,v) = M(u,v)
\]
for all $(u,v) \in [0,1]^2$. The upper bound $M$ is the comonotone coupling; the lower bound $W$ is the countermonotone coupling.
\end{theorem}

This is the master inequality. Every other bound in this paper is a constrained version of Theorem~\ref{thm:frechet}. The comonotone coupling $M$ corresponds to rank-preserving assignment (the best case for treatment effects when better units select into treatment); the countermonotone $W$ corresponds to rank-reversing assignment.

\subsection{The Bell--Boole Inequality as a Fr\'{e}chet Bound}

Consider three $\pm 1$-valued random variables $A, B, C$ jointly distributed on a classical probability space. The original \emph{Boole--Bell inequality} \cite{boole1854investigation, bell1964einstein} states
\begin{equation}\label{eq:boole-bell}
|E[AB] - E[AC]| \leq 1 - E[BC].
\end{equation}
This is a necessary condition for a joint distribution to exist. \citet{bell1964einstein} showed quantum mechanics violates this for measurements on entangled spin-1/2 particles---the correlations cannot be realized by any coupling of the marginal $\pm 1$ distributions.

\begin{proposition}[Probabilistic reformulation]\label{prop:bell-frechet}
Bell's inequality is equivalent to the condition that the $3 \times 3$ correlation matrix
\[
\Sigma = \begin{pmatrix} 1 & E[AB] & E[AC] \\ E[AB] & 1 & E[BC] \\ E[AC] & E[BC] & 1 \end{pmatrix}
\]
is the covariance matrix of some joint distribution on $\{-1,+1\}^3$. Violation is equivalent to $\Sigma$ lying outside the Fr\'{e}chet-feasible set for $\pm 1$ variables.
\end{proposition}

This makes precise the sense in which Bell violation is a Fr\'{e}chet infeasibility---an ``impossible coupling'' in the terminology of \citet{gill2003accardi}. The quantum correlation matrix $\Sigma_Q$ (achievable by entangled qubits) is positive semidefinite but fails the constraint that it be the second moment matrix of a $\{-1,+1\}^3$-valued random vector.

\subsection{Fine's Theorem: The Pivot}\label{sec:fine}

The following characterization unifies the Bell and statistical perspectives.

\begin{theorem}[Fine, 1982]\label{thm:fine}
For a quantum correlation experiment with observables $A, A', B, B' \in \{-1,+1\}$, the following are equivalent:
\begin{enumerate}
\item There is a deterministic local hidden variable model.
\item There is a factorizable stochastic model.
\item There is a single joint distribution for all four observables returning the experimental probabilities.
\item The CHSH inequalities hold: $|E[AB] + E[AB'] + E[A'B] - E[A'B']| \leq 2$.
\end{enumerate}
\end{theorem}

Fine's theorem says that Bell's locality condition \emph{is} the Fr\'{e}chet joint-distribution-existence condition for $\pm 1$ variables. This is the pivot connecting quantum physics to classical probability (see \cite{gill2014statistics} for a thorough statistical treatment).

\subsection{The CHSH Polytope}\label{sec:chsh}

In the CHSH scenario \cite{clauser1969proposed}, Alice and Bob each choose between two measurement settings, obtaining $\pm 1$ outcomes. The space of all possible correlation tables $\{E[AB], E[AB'], E[A'B], E[A'B']\}$ is 4-dimensional. The \emph{local-realist polytope} is the convex hull of the 16 deterministic strategies (4 bits of pre-agreed answers). Its facets are the CHSH inequalities:
\begin{equation}\label{eq:chsh}
|\langle AB \rangle + \langle AB' \rangle + \langle A'B \rangle - \langle A'B' \rangle| \leq 2.
\end{equation}
Quantum mechanics achieves values up to $2\sqrt{2}$ (the Tsirelson bound), lying strictly outside this polytope but inside the no-signaling polytope (bound of 4 \cite{popescu1994quantum}). The geometry is (see \cite{brunner2014bell} for a comprehensive review)
\[
\text{Local polytope} \subset \text{Quantum set} \subset \text{No-signaling polytope}.
\]
Figure~\ref{fig:polytope} illustrates this containment in a cross-section of the correlation space. The local polytope (inner square) has facets at $|S|=2$; the quantum set (inscribed circle) reaches $|S|=2\sqrt{2}$; the no-signaling polytope (outer square) has facets at $|S|=4$.

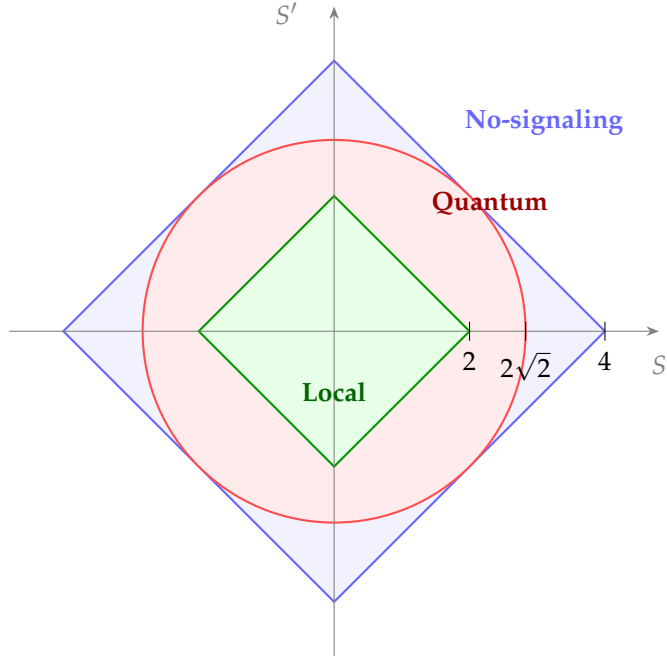
\begin{figure}[H]
\centering
\begin{tikzpicture}[scale=1.8]
  \fill[blue!5] (0,2) -- (2,0) -- (0,-2) -- (-2,0) -- cycle;
  \draw[blue!60, thick] (0,2) -- (2,0) -- (0,-2) -- (-2,0) -- cycle;
  \fill[red!8] (0,0) circle ({sqrt(2)});
  \draw[red!70, thick] (0,0) circle ({sqrt(2)});
  \fill[green!10] (0,1) -- (1,0) -- (0,-1) -- (-1,0) -- cycle;
  \draw[green!60!black, thick] (0,1) -- (1,0) -- (0,-1) -- (-1,0) -- cycle;
  \draw[gray, -Stealth] (-2.4,0) -- (2.4,0);
  \draw[gray, -Stealth] (0,-2.4) -- (0,2.4);
  \draw (1,0.07) -- (1,-0.07) node[below, font=\footnotesize] {$2$};
  \draw ({sqrt(2)},0.07) -- ({sqrt(2)},-0.07) node[below, font=\footnotesize] {$2\sqrt{2}$};
  \draw (2,0.07) -- (2,-0.07) node[below, font=\footnotesize] {$4$};
  \node[green!40!black, font=\footnotesize\bfseries] at (0,-0.45) {Local};
  \node[red!60!black, font=\footnotesize\bfseries] at (1.15,0.95) {Quantum};
  \node[blue!60, font=\footnotesize\bfseries] at (1.55,1.55) {No-signaling};
  \node[gray, font=\footnotesize] at (2.4, -0.25) {$S$};
  \node[gray, font=\footnotesize] at (-0.35, 2.35) {$S'$};
\end{tikzpicture}
\caption{Cross-section of the correlation space for the CHSH scenario, showing the three nested convex bodies. The local-realist polytope (inner diamond, $|S| \leq 2$) is contained in the quantum set (circle, $|S| \leq 2\sqrt{2}$), which is contained in the no-signaling polytope (outer diamond, $|S| \leq 4$). Here $S = \langle AB\rangle + \langle AB'\rangle + \langle A'B\rangle - \langle A'B'\rangle$ and $S' = \langle AB\rangle - \langle AB'\rangle + \langle A'B\rangle + \langle A'B'\rangle$ are two independent CHSH combinations.}\label{fig:polytope}
\end{figure}

This three-layer structure has a direct causal analog. The local polytope corresponds to causal models with classical (factorable) confounding satisfying the IV structural assumptions; effects are partially identified within the Balke--Pearl bounds. The quantum set corresponds to models where the confounder is quantum-mechanical, yielding a wider identified set. The no-signaling polytope corresponds to models with arbitrary confounding and no structural assumptions---Manski's worst-case width-1 bounds. Structural assumptions move the analyst from the outer layer inward, tightening the feasible set.

\section{Causal Inference as Polytope Geometry}\label{sec:causal}

The polytope framework of Section~\ref{sec:polytope} has direct counterparts in causal inference, where the fundamental problem is the same: bounding a quantity (treatment effect, probability of causation) that depends on a joint distribution that is never fully observed.

\subsection{Pearl's Instrumental Inequality}

\begin{definition}\label{def:iv}
The canonical instrumental variable (IV) model consists of an instrument $Z$, treatment $X$, outcome $Y$, and latent confounder $U$, with the causal structure $Z \to X \to Y$, $U \to X$, $U \to Y$, and $Z \perp U$.
\end{definition}

\begin{theorem}[Pearl, 1995]\label{thm:pearl-iv}
Under the IV model, for each $x \in \mathcal{X}$:
\begin{equation}\label{eq:pearl-iv}
\max_z \sum_y \max_{x'} P(Y=y, X=x' \mid Z=z) \leq 1.
\end{equation}
This is a necessary testable implication of the IV structural assumptions.
\end{theorem}

This is the \emph{causal Bell inequality}: a constraint on the observed distribution that must hold if the IV structure is valid. Violation certifies model misspecification, exactly as Bell violation certifies the absence of a local hidden variable.

\begin{theorem}[Structural equivalence; \cite{pearl1995testability, balke1997bounds}]\label{thm:structural-equiv}
For binary $X, Y, Z$, the following are equivalent:
\begin{enumerate}
\item The IV model holds.
\item A factorizable hidden-variable model $P(X,Y \mid Z,U) = P(X \mid Z,U)\,P(Y \mid X,U)$ exists.
\item Pearl's instrumental inequality holds for all $x$.
\end{enumerate}
Moreover, when these hold, the Balke--Pearl bounds are non-trivially informative (width $< 1$). For non-binary variables, item 3 is necessary but may not be sufficient for items 1--2; additional inequality constraints arise from higher-order marginal compatibility.
\end{theorem}

The parallel to Fine's theorem (Theorem~\ref{thm:fine}) is exact: IV validity $\leftrightarrow$ local realism; instrument $Z \leftrightarrow$ measurement setting; treatment $X \leftrightarrow$ particle $A$; outcome $Y \leftrightarrow$ particle $B$; confounder $U \leftrightarrow$ hidden variable $\lambda$.

\subsection{Balke--Pearl Bounds via Linear Programming}

Following \citet{balke1997bounds}, enumerate \emph{response types} for each unit. For binary $X, Y, Z$, there are 4 response types for $X$ given $Z$ and 4 for $Y$ given $X$, yielding 16 possible unit types $\{t_{ij}\}$ with probabilities $\{q_{ij}\}$.

The observed distribution constrains $q$ via linear equalities:
\begin{equation}\label{eq:response-type}
P(Y=y, X=x \mid Z=z) = \sum_{\substack{i,j:\, \text{type}(i,j) \\ \text{generates } (x,y \mid z)}} q_{ij}.
\end{equation}
The \emph{average causal effect} (ACE) is linear in $q$:
\begin{equation}\label{eq:ace}
\mathrm{ACE} = \sum_{i,j} c_{ij}\, q_{ij}.
\end{equation}

\begin{theorem}[Balke--Pearl, 1997]\label{thm:balke-pearl}
Tight lower and upper bounds on the ACE are
\[
[\mathrm{ACE}]_{\min} = \min_{q \geq 0,\, \mathbf{A}q = p} c^T q, \qquad [\mathrm{ACE}]_{\max} = \max_{q \geq 0,\, \mathbf{A}q = p} c^T q,
\]
where $\mathbf{A}$ is the constraint matrix mapping response-type probabilities to observed distributions, and $p$ is the observed distribution vector. The bounds are sharp and obtained by linear programming.
\end{theorem}

The feasible set $\{q : \mathbf{A}q = p,\; q \geq 0\}$ is the \emph{response-function polytope}---the causal analog of the Fr\'{e}chet class. Its extreme points are the deterministic strategies, and its facets give the causal bounds.

\subsection{Tian--Pearl: Probabilities of Causation}

Beyond average effects, \citet{tian2000probabilities} bound individual-level causal quantities.

\begin{definition}\label{def:pns}
For binary treatment $X$ and outcome $Y$:
\begin{itemize}
\item \textbf{PN} (probability of necessity): $\mathrm{PN} = P(Y_{x'} = 0 \mid X=1, Y=1)$---the probability that treatment was necessary for the observed outcome.
\item \textbf{PS} (probability of sufficiency): $\mathrm{PS} = P(Y_x = 1 \mid X=0, Y=0)$---the probability that treatment would have been sufficient.
\item \textbf{PNS} (probability of necessity and sufficiency): $\mathrm{PNS} = P(Y_x = 1, Y_{x'} = 0)$.
\end{itemize}
\end{definition}

\begin{theorem}[Tian--Pearl, 2000]\label{thm:tian-pearl}
Under no unmeasured confounding, sharp bounds on PNS using both experimental data $P(y_x), P(y_{x'})$ and observational data $P(y \mid x)$ are:
\begin{align*}
\max\bigl\{0,\; P(y_x) - P(y_{x'}),\; P(y) - P(y_{x'}),\; P(y_x) - P(y)\bigr\} &\leq \mathrm{PNS} \\
\mathrm{PNS} &\leq \min\bigl\{P(y_x),\; P(y'_{x'}),\; P(x,y) + P(x',y')\bigr\}.
\end{align*}
The bounds are sharp: every combination of data satisfying these inequalities is achievable by some causal model.
\end{theorem}

\subsection{Manski's Partial Identification}

\citet{manski1990nonparametric} arrived at the same polytope from the potential outcomes tradition.

\begin{theorem}[Manski, 1990]\label{thm:manski}
Under no assumptions beyond bounded outcomes $Y \in [0,1]$, the average treatment effect $\mathrm{ATE} = E[Y(1)] - E[Y(0)]$ satisfies:
\begin{align*}
&E[Y \mid X=1]\,P(X=1) + 0 \cdot P(X=0) \\
&\quad - \bigl[E[Y \mid X=0]\,P(X=0) + 1 \cdot P(X=1)\bigr] \\
&\leq \mathrm{ATE} \\
&\leq E[Y \mid X=1]\,P(X=1) + 1 \cdot P(X=0) \\
&\quad - \bigl[E[Y \mid X=0]\,P(X=0) + 0 \cdot P(X=1)\bigr].
\end{align*}
This interval has width 1 and always contains zero.
\end{theorem}

Under the IV monotonicity assumption, the Balke--Pearl, Manski, and Heckman--Vytlacil bounds coincide---all characterize the same Fr\'{e}chet polytope under the same structural constraints \cite{heckman2005structural}. The causal inference literature has developed a family of monotonicity conditions---monotone treatment response, monotone treatment selection, and monotone instrumental variables \cite{manski2000monotone}---each of which cuts a face from the feasibility polytope and tightens the bounds. These are the causal analogs of closing loopholes in Bell experiments \cite{larsson2004bell}: each structural assumption (monotonicity, exclusion restriction, independence) removes part of the feasible polytope, just as each physical assumption (locality, fair sampling, freedom of choice) tightens the bound on achievable classical correlations.

\section{Kolmogorov Complexity and Entropic Bell Inequalities}\label{sec:kolmogorov}

The polytope framework of the preceding sections is formulated in the language of probability distributions. Kolmogorov complexity provides an alternative, algorithmic formulation that makes no reference to probability, applies to individual strings, and illuminates Bell violation from the perspective of computability and information compression. This section is independent of Sections~\ref{sec:quantum} and~\ref{sec:qbc}; it offers parallel insight into the same phenomenon rather than a logical foundation for quantum speedups.

\subsection{The Algorithmic Lens}

\begin{definition}\label{def:kc}
The Kolmogorov complexity $K(x)$ of a binary string $x$ is the length of the shortest program $p$ on a universal Turing machine $U$ such that $U(p) = x$. The conditional complexity $K(x \mid y)$ is the shortest program producing $x$ given $y$ as input.
\end{definition}

The fundamental link between Kolmogorov complexity and Shannon entropy is that for any random variable $X$ drawn from a computable distribution $P$, the expected prefix-free complexity satisfies $E[K(X)] = H(X) + O(1)$, where the $O(1)$ term depends on the choice of universal machine. (This equality requires the prefix-free variant of $K$; for the plain complexity of Definition~\ref{def:kc}, only the upper bound $E[K(X)] \leq H(X) + O(1)$ holds.) Kolmogorov complexity is thus the \emph{individual} version of Shannon entropy, defined without reference to a probability measure.

\subsection{The Algorithmic Bell Inequality}

\begin{theorem}[Algorithmic Bell; \cite{berthiaume2001quantum, vitanyi2000quantum}]\label{thm:algo-bell}
Under local realism (hidden variable $\lambda$):
\begin{equation}\label{eq:kc-bell}
K(AB \mid XY\lambda) \leq K(A \mid X\lambda) + K(B \mid Y\lambda) + O(\log K(\lambda)).
\end{equation}
Quantum entanglement violates this: for maximally entangled qubits, the joint outcomes $AB$ have higher conditional Kolmogorov complexity given any classical seed $\lambda$ than any locally-computable approximation permits.
\end{theorem}

This is the algorithmic version of Bell's theorem. A local hidden variable model must compress the joint measurement record $(AB)$ into the sum of locally-computable parts. Entanglement means no such decomposition exists.

\noindent\textbf{Remark (QKD security).}
Quantum key distribution is secure precisely because the joint outcomes of a Bell-violating experiment cannot be efficiently compressed by any eavesdropper holding a classical program for $\lambda$. Bell violation certifies that no local hidden variable (and hence no eavesdropper's side information) can reproduce the observed correlations.

\subsection{The Kakeya Connection}

\citet{lutz2017algorithmic} proved the 2D Kakeya conjecture using conditional Kolmogorov complexity: every plane set containing a unit line segment in every direction has Hausdorff dimension 2. The proof exploits the \emph{point-to-set principle}: if every point $x$ in a set $E$ has high conditional complexity $K(x \mid d) \approx 2$, where $d$ is the direction of the line segment through $x$, then $E$ has full dimension.

The structural parallel to Bell is suggestive: a set with a unit line in every direction cannot have low conditional complexity---it cannot be ``locally explained'' by the direction alone, just as Bell-violating outcomes cannot be locally explained by the hidden variable $\lambda$. The analogy is formal in the sense that both results involve lower bounds on conditional information given a parameter, though the mathematical settings differ.

\begin{theorem}[Lutz--Lutz, 2017]\label{thm:kakeya}
If $E \subseteq \mathbb{R}^2$ contains a unit line segment in every direction $\theta \in [0,\pi)$, then for every $x \in E$:
\[
K(x \mid \theta) \geq 2 - \epsilon
\]
up to an additive constant. Hence $\dim_H(E) = 2$.
\end{theorem}

\noindent\textbf{Analogy (not a formal connection):} The structural parallel to Bell is suggestive: ``every direction'' resembles ``every measurement setting''; ``locally explained by direction'' resembles ``locally explained by $\lambda$''; ``information content remains high despite constraints'' resembles ``outcomes remain incompressible despite hidden variables.'' However, the Kakeya result and Bell's theorem involve fundamentally different objects (geometric sets vs. probability distributions) and different notions of complexity (Hausdorff dimension vs. correlation polytopes). The analogy illuminates the principle of non-factorizability but should not be read as a formal equivalence.

\subsection{Entropic Bell Inequalities}

\begin{definition}[Shannon cone]\label{def:shannon-cone}
For $n$ random variables $X_1, \ldots, X_n$, the \emph{Shannon cone} $\Gamma^*_n$ is the closure of the set of all achievable entropy vectors $\mathbf{h} = (H(X_S))_{S \subseteq [n]}$ under classical joint distributions.
\end{definition}

\begin{theorem}[\cite{chaves2014causal}]\label{thm:entropic-bell}
If a Bell experiment is locally realistic, its entropy vector lies in $\Gamma^*_n$ restricted by the causal structure of the Bell DAG. Quantum systems can produce entropy vectors outside this cone.
\end{theorem}

The entropic constraints derived from the Bell DAG take the general form of linear inequalities on the entropy vector. For example, the CHSH-type entropic inequality (a consequence of the submodularity of Shannon entropy applied to the Bell causal structure; see \cite{chaves2014causal}, Section~III) is:
\begin{equation}\label{eq:entropic-chsh}
I(A:B) + I(A:B') + I(A':B) - I(A':B') \leq 2\,H(\text{settings}).
\end{equation}
Quantum systems can violate entropic constraints of this type by exploiting negative conditional von Neumann entropy, a signature of entanglement with no classical causal explanation.

\subsection{The Shannon--Kolmogorov--Quantum Arc}

The three levels of information theory form a hierarchy, presented in Table~\ref{tab:info-arc}.

\begin{table}[H]
\centering
\caption{The three levels of information theory and their Bell connections.}\label{tab:info-arc}
\small
\begin{tabular}{@{}lll@{\hspace{6pt}}l@{}}
\toprule
Level & Description & Quantity & Bell connection \\
\midrule
Shannon & Probabilistic & $H(X) = -\!\sum p(x)\log p(x)$ & Entropic Bell ineq. \\
Kolmogorov & Algorithmic & $K(x)$, shortest program & Algorithmic Bell thm. \\
von Neumann & Quantum & $S(\rho) = -\tr(\rho \log \rho)$ & Quantum entropic ineq. \\
\bottomrule
\end{tabular}
\end{table}

The classical Brudno theorem \cite{brudno1983entropy} connects Shannon and Kolmogorov levels (the entropy rate of an ergodic source equals the asymptotic Kolmogorov complexity rate); the quantum Brudno theorem \cite{benatti2006quantum} connects Kolmogorov and quantum levels. The Photons = Tokens framework \cite{litowitz2026photons} connects the physical cost at the quantum level to the economic cost at the Shannon level via Landauer's principle.

\section{Quantum Bell Inequalities}\label{sec:quantum}

Sections~\ref{sec:polytope}--\ref{sec:kolmogorov} characterize the classical boundary---the set of correlations achievable by any joint probability distribution. Quantum mechanics exceeds this boundary. The physical mechanism is non-commutativity: quantum observables are operators on a Hilbert space, and when two observables do not commute, they cannot be simultaneously diagonalized, and no joint probability distribution over their outcomes exists. This is precisely the condition that Fine's theorem identifies as equivalent to Bell violation.

This section develops the operator-algebraic extension: the Tsirelson bound, the NPA hierarchy for computational approximation, and the MIP$^*$=RE theorem establishing the fundamental undecidability of the quantum boundary. From the causal inference perspective, the quantum extension answers the question: if the latent confounder $U$ in an IV model were a quantum system rather than a classical random variable, how much wider could the range of achievable correlations be?

\subsection{The Tsirelson Bound}

\begin{theorem}[Tsirelson, 1980]\label{thm:tsirelson}
The maximum quantum violation of the CHSH inequality is $2\sqrt{2}$:
\[
\sup_{\rho,\, A,\, A',\, B,\, B'} \bigl|\langle AB \rangle + \langle AB' \rangle + \langle A'B \rangle - \langle A'B' \rangle\bigr| = 2\sqrt{2},
\]
where the supremum is over all quantum states $\rho$ and measurement operators with eigenvalues $\pm 1$.
\end{theorem}

\begin{proof}
Let $A, A'$ be $\pm 1$-valued observables on Alice's Hilbert space $\mathcal{H}_A$ and $B, B'$ on Bob's $\mathcal{H}_B$. Operators on different subsystems commute: $[A \otimes I,\, I \otimes B] = 0$, but $A, A'$ on Alice's system and $B, B'$ on Bob's need not commute---and it is precisely this non-commutativity that enables Bell violation. Writing $\mathcal{C} = A \otimes (B+B') + A' \otimes (B-B')$ and expanding $\mathcal{C}^2$ using $A^2 = A'^2 = B^2 = B'^2 = I$, cross-terms cancel by the tensor product structure: $(B+B')^2 + (B-B')^2 = 2\{B,B'\} + 2(2I - \{B,B'\}) = 4I$ and $(B+B')(B-B') + (B-B')(B+B') = 2[B,B']$. Thus
\[
\mathcal{C}^2 = 4I \otimes I - [A, A'] \otimes [B, B'].
\]
Since $\|[A,A']\|_{\mathrm{op}} \leq 2$ and $\|[B,B']\|_{\mathrm{op}} \leq 2$, we obtain $\|\mathcal{C}^2\|_{\mathrm{op}} \leq 8$, hence $\|\mathcal{C}\|_{\mathrm{op}} \leq 2\sqrt{2}$. The bound is achieved by anticommuting observables on maximally entangled qubits.
\end{proof}

\subsection{The NPA Hierarchy}

For bounding quantum correlations computationally, \citet{navascues2007bounding, navascues2008convergent} developed a hierarchy of semidefinite programs (SDPs).

\begin{theorem}[NPA hierarchy]\label{thm:npa}
For each level $k \geq 1$, the level-$k$ NPA bound $\beta_k$ satisfies
\[
\beta_1 \geq \beta_2 \geq \cdots \geq \beta_\infty = T_c,
\]
where $T_c$ is the commuting Tsirelson bound. Each $\beta_k$ is computable via an SDP of size $O(n^{2k})$.
\end{theorem}

The construction is as follows. Let $\{O_i\}$ be the set of measurement operators (including products up to length $k$). A \emph{moment matrix} $\Gamma$ has entries $\Gamma_{ij} = \tr(\rho\, O_i^\dagger O_j)$. Any quantum realization produces a moment matrix that is positive semidefinite and satisfies linear constraints from the algebraic relations $O_i^2 = I$ (projective measurements) and $[O_i^A, O_j^B] = 0$ (operators on different subsystems commute). The level-$k$ NPA relaxation maximizes a Bell functional over all PSD matrices $\Gamma \succeq 0$ satisfying these constraints:
\begin{equation}\label{eq:npa-sdp}
\beta_k = \max_\Gamma \sum_{x,y,a,b} \mathcal{B}_{xyab}\, \Gamma_{(a|x),(b|y)} \quad \text{s.t.} \quad \Gamma \succeq 0,\; \text{algebraic constraints at level } k,
\end{equation}
where $\mathcal{B}_{xyab}$ are the coefficients of the Bell functional being optimized. This is a standard SDP and can be solved efficiently for small $k$. The parallel with causal inference is direct: the Balke--Pearl LP maximizes a causal functional (the ACE) over the response-function polytope $\{q : \mathbf{A}q = p, q \geq 0\}$. The LP is the classical (commutative) version; the SDP is the quantum (non-commutative) extension. Both are outer relaxations that converge to the true feasible set---the LP is exact for binary variables, and the NPA hierarchy converges to $C_{qc}$ in the limit. \citet{gill2007better} showed how statistical tools from missing-data maximum likelihood can be used to optimize Bell inequalities, providing a further connection between statistical methodology and quantum correlation bounds.

\subsection{MIP$^*$ = RE and the Undecidability of Quantum Correlations}

The most striking recent result is the resolution of Tsirelson's problem.

\begin{theorem}[Ji--Natarajan--Vidick--Wright--Yuen, 2020]\label{thm:mip-re}
$\mathrm{MIP}^* = \mathrm{RE}$. Consequently:
\begin{enumerate}
\item The tensor-product quantum correlation set $C_{qs}$ is not closed.
\item The commuting-operator correlation set $C_{qc}$ strictly contains $C_{qs}$.
\item Tsirelson's problem has a negative answer: $C_{qs} \neq C_{qc}$.
\item Connes' embedding problem has a negative answer.
\item Determining whether a correlation table is quantum-achievable is undecidable.
\end{enumerate}
\end{theorem}

This is the Kolmogorov complexity statement: the quantum correlation set is \emph{not computable}. No finite algorithm exists to decide whether a given correlation table lies in the quantum set---the problem is RE-complete, at least as hard as the halting problem \cite{ji2021mip}.

\subsection{Physical Principles Bounding the Tsirelson Value}

Why $2\sqrt{2}$ and not 4 (the no-signaling maximum)? Three equivalent characterizations provide the answer.

\emph{Information causality} \cite{pawlowski2009information}: the information Bob can gain about Alice's data cannot exceed the classical bits she sends. Violation of information causality occurs if and only if CHSH $> 2\sqrt{2}$.

\emph{Macroscopic locality} \cite{navascues2010glance}: coarse-grained averages of quantum experiments must satisfy classical probability. Violation if and only if CHSH $> 2\sqrt{2}$.

\emph{Local quantum mechanics} \cite{oppenheim2010uncertainty}: the Heisenberg uncertainty principle in local measurement bases is saturated at the Tsirelson bound.

Each of these is an information-theoretic constraint on the feasibility polytope: they describe the boundary of the quantum set from three different directions, just as MTR, MTS, and MIV describe the boundary of the causal identified set from three different directions in Manski's framework.

\section{Quantum Bayesian Computation}\label{sec:qbc}

The preceding sections establish the polytope framework and its quantum extension. We now show that the same non-commutativity responsible for Bell violation provides computational advantage for Bayesian inference.

\subsection{The Born Rule--Bayes Rule Duality}

The central insight of QBC \cite{polson2023quantum} is that Bayesian inference and quantum measurement are the same operation in different algebraic settings.

\emph{Classical Bayes:}
\begin{equation}\label{eq:bayes}
\pi(\theta \mid x) = \frac{p(x \mid \theta)\,\pi(\theta)}{p(x)}, \qquad p(x) = \int p(x \mid \theta)\,\pi(\theta)\,d\theta.
\end{equation}

\emph{Quantum Born} (in this section, $\mathsf{M}$ denotes a measurement operator, distinct from the Fr\'{e}chet upper bound $M$ of Section~\ref{sec:polytope}):
\begin{equation}\label{eq:born}
\rho' = \frac{\mathsf{M}\rho \mathsf{M}^\dagger}{\tr(\mathsf{M}\rho \mathsf{M}^\dagger)}, \qquad p(x) = \tr(\mathsf{M}_x \rho).
\end{equation}

The correspondence is: density matrix $\rho \leftrightarrow$ prior $\pi(\theta)$; measurement operator $\mathsf{M} \leftrightarrow$ likelihood $p(x \mid \theta)$; post-measurement state $\rho' \leftrightarrow$ posterior $\pi(\theta \mid x)$; Born probability $\leftrightarrow$ marginal $p(x)$.

\begin{proposition}[QBC duality]\label{prop:qbc}
When the density matrix $\rho$ is diagonal in the computational basis, quantum measurement reduces to classical Bayesian conditioning: $\rho'_{ii} = \pi(\theta_i \mid x)$. Non-commutativity of the measurement operators with $\rho$---which in multipartite systems manifests as entanglement---is the departure from classical Bayesian inference.
\end{proposition}

\begin{proof}
When $\rho$ is diagonal with $\rho_{ii} = \pi(\theta_i)$ and $\mathsf{M}_j$ is a diagonal projector with $(\mathsf{M}_j)_{ii} = p(x_j \mid \theta_i)$, then $p(x_j) = \tr(\mathsf{M}_j \rho) = \sum_i p(x_j \mid \theta_i)\,\pi(\theta_i)$ (the marginal likelihood), and $(\mathsf{M}_j \rho \mathsf{M}_j)_{ii}/p(x_j) = p(x_j \mid \theta_i)\,\pi(\theta_i)/p(x_j) = \pi(\theta_i \mid x_j)$ by Bayes' theorem. Non-commutativity of $\mathsf{M}_j$ with $\rho$ corresponds to quantum coherence.
\end{proof}

\subsection{Bell Violation as Bayesian Non-Classicality}

\begin{theorem}\label{thm:bell-bayes}
A Bell-violating correlation table $P(A,B \mid X,Y)$ admits no classical Bayesian model: there is no prior $\pi(\lambda)$ and likelihoods $p(A \mid X,\lambda)$, $p(B \mid Y,\lambda)$ such that
\[
P(A,B \mid X,Y) = \int p(A \mid X,\lambda)\,p(B \mid Y,\lambda)\,d\pi(\lambda).
\]
This is equivalent to Fine's theorem: the joint distribution $\pi(A,B,\lambda \mid X,Y)$ does not exist in the Kolmogorov sense.
\end{theorem}

\begin{proof}
By Fine's theorem (Theorem~\ref{thm:fine}), a local hidden variable model exists if and only if a joint distribution on $(A,A',B,B')$ exists. If such a joint distribution existed, the CHSH inequality would hold. Bell violation ($|\cdot| > 2$) therefore implies no such distribution and hence no Bayesian model $\pi(\lambda)\,p(A \mid X,\lambda)\,p(B \mid Y,\lambda)$.
\end{proof}

Bell violation is thus equivalent to the non-existence of a classical Bayesian model for quantum correlations. QBC exploits this by using quantum states (density matrices) as priors---they can represent distributions that have no classical joint representation, enabling computation beyond what classical MCMC achieves.

\subsection{Quantum Computational Speedups for Bayesian Inference}

The connection between Bell violation and quantum speedups in computation is via Theorem~\ref{thm:bell-bayes}: a Bell-violating correlation structure admits no classical joint distribution. More generally, any quantum state with non-diagonal density matrix (coherence/entanglement) encodes correlations that cannot be reproduced by any classical measurement-by-measurement simulation. This entanglement is the computational resource: an algorithm that exploits it gains speedup precisely because the intermediate states violate classical feasibility constraints.

Three results illustrate this principle:

\emph{Quantum Gaussian Processes.} For a Gaussian process with kernel matrix $K \in \mathbb{R}^{N \times N}$, the posterior mean $\mu^* = K(K + \sigma^2 I)^{-1}y$ requires $O(N^3)$ classically. The HHL algorithm \cite{harrow2009quantum} achieves $O(\kappa \log N / \varepsilon)$ where $\kappa$ is the condition number, by maintaining quantum superposition over matrix-vector product evaluations. This superposition---if projected onto a measurement basis---would violate Bell inequalities across different matrix-entry subproblems, making sequential classical computation necessary.

\emph{Quantum MCMC.} Quantum walks on the Markov chain state space achieve quadratic speedup: $O(\sqrt{1/\delta})$ steps versus $O(1/\delta)$ classically \cite{szegedy2004quantum}. The speedup requires coherent superposition over the transition graph, maintained throughout the walk. Classical simulation would require exponential overhead to represent all ``simultaneously explored'' paths.

\emph{Quantum Stochastic Gradient Descent.} Amplitude estimation \cite{montanaro2015quantum} yields $O(1/\varepsilon)$ convergence versus $O(1/\varepsilon^2)$ classically for gradient estimation, via superposition over objective evaluations. Each speedup arises from maintaining quantum coherence (non-commutativity between successive operations), which would be destroyed by intermediate classical measurement.

\subsection{K-GAM as Classical Approximation to Quantum QBC}

\citet{polson2025kgam} show that Kolmogorov's Superposition Theorem provides the classical architecture closest in spirit to quantum function evaluation.

\begin{theorem}[\cite{kolmogorov1957representation}]\label{thm:kst}
For any continuous $f: [0,1]^n \to \mathbb{R}$, there exist continuous $\Phi_q: \mathbb{R} \to \mathbb{R}$ and $\varphi_{q,p}: [0,1] \to \mathbb{R}$ such that
\begin{equation}\label{eq:kst}
f(x_1, \ldots, x_n) = \sum_{q=0}^{2n} \Phi_q\!\left(\sum_{p=1}^{n} \varphi_{q,p}(x_p)\right).
\end{equation}
\end{theorem}

The K-GAM network architecture implements this theorem computationally. Both QBC and K-GAM share the goal of representing complex posterior distributions via sparse, factorized structures. Table~\ref{tab:kgam-qbc} presents suggestive parallels between the two approaches:

\begin{table}[H]
\centering
\caption{Structural parallels between QBC and K-GAM (analogies, not formal equivalences).}\label{tab:kgam-qbc}
\small
\begin{tabular}{@{}p{0.38\textwidth}p{0.55\textwidth}@{}}
\toprule
QBC concept & K-GAM analog \\
\midrule
Density matrix (encodes correlations) & Embedding functions $\varphi_{q,p}$ (encode feature interactions) \\
Non-diagonal terms (coherence) & Outer functions $\Phi_q$ summed over indices (nonlinearity) \\
Born rule: condition via projector & Prediction: $f(x) = \sum_q \Phi_q(\sum_p \varphi_{q,p}(x_p))$ \\
Quantum superposition of amplitudes & $O(n)$ functions vs.\ MLP $O(2^n)$ neurons (KST bound) \\
\bottomrule
\end{tabular}
\end{table}

\noindent\textbf{Status of this correspondence:} The parallels in Table~\ref{tab:kgam-qbc} are inspirational but not formally proven. K-GAM is a classical architecture that avoids the combinatorial explosion of fully connected networks; it achieves this through additive structure, not through quantum principles. Whether K-GAM can be shown to achieve the same information-theoretic bounds as QBC, or whether it serves as an optimal classical ``shadow'' of quantum computation in a rigorous sense, remains open.

\begin{proposition}[Horseshoe sparsity in K-GAM]\label{prop:horseshoe}
Consider the K-GAM regression model:
\begin{equation}\label{eq:kgam-horseshoe}
y_i = \sum_{q=0}^{2n} \Phi_q\!\left(\sum_{p=1}^{n} \varphi_{q,p}(x_{ip})\right) + \epsilon_i, \quad \epsilon_i \sim N(0, \sigma^2),
\end{equation}
where the outer functions $\Phi_q$ are parameterized by coefficients $\theta_q \in \mathbb{R}^d$, each equipped with the horseshoe prior \cite{polson2010shrink, polson2012half}:
\[
\theta_{q,j} \mid \lambda_j, \tau \sim N(0, \lambda_j^2 \tau^2), \quad \lambda_j \sim \mathrm{C}^+(0,1), \quad \tau \sim \mathrm{C}^+(0,1).
\]
Here $\lambda_j$ is the local shrinkage parameter and $\tau$ is the global shrinkage parameter.
Then:
\begin{enumerate}
\item \textbf{(Posterior Sparsity)} The posterior $\pi(\theta \mid \mathcal{D})$ concentrates on solutions where the number of non-negligible coefficients is strictly smaller than $d(2n+1)$.
\item \textbf{(Near-minimax rates)} For signals $\theta^* \in \mathcal{S}_s$ (the set of $s$-sparse vectors), the posterior mean $\hat{\theta}$ satisfies
\[
\mathbb{E}_{\theta^*} \|\hat{\theta} - \theta^*\|_2^2 \leq C\,s \log(dn) / N
\]
up to constants depending on the noise level $\sigma^2$ and feature scales.
\item \textbf{(Model selection)} The posterior mode automatically selects a subset of outer functions $\Phi_q$ needed to fit the data, discarding functions whose coefficients are pushed to zero. The effective number of active functions is of order $s$, strictly less than the KST upper bound of $2n+1$.
\end{enumerate}
\end{proposition}

\begin{proof}
We prove each item.

\noindent\textit{Item 1 (Posterior Sparsity).} The horseshoe marginal, obtained by integrating $\lambda_j$ out of $N(0, \lambda_j^2 \tau^2) \cdot C^+(0,1)$, satisfies \cite{polson2010shrink}:
\[
\pi(\theta_{q,j} \mid \tau) = \int_0^\infty \frac{1}{\sqrt{2\pi \lambda^2 \tau^2}} \exp\!\left(-\frac{\theta_{q,j}^2}{2\lambda^2 \tau^2}\right) \frac{2}{\pi(1+\lambda^2)}\,d\lambda.
\]
This density has two key properties: (i) a pole at the origin ($\pi(\theta \mid \tau) \to \infty$ as $\theta \to 0$), concentrating mass near zero; and (ii) Cauchy-like tails $\pi(\theta \mid \tau) \sim \log(1 + 4\tau^2/\theta^2)$ for large $|\theta|$, which do not over-shrink large signals. By Bayes' theorem,
\[
\pi(\theta \mid \mathcal{D}) \propto L(\mathcal{D} \mid \theta) \prod_{q,j} \pi(\theta_{q,j} \mid \tau),
\]
where $L(\mathcal{D} \mid \theta) = \exp\bigl(-\frac{1}{2\sigma^2} \sum_i (y_i - f(x_i; \theta))^2\bigr)$. The pole at zero drives most coefficients to negligible values while the heavy tails preserve large signals. Polson \& Scott (2010, Theorem 2) show that a fixed number of coefficients dominate the posterior mean, with the rest shrunk to near zero.

\noindent\textit{Item 2 (Near-minimax rates).} The horseshoe prior achieves near-minimax rates for sparse estimation in the standard Gaussian sequence model (Polson \& Scott, 2010, Theorem 3). This result extends to high-dimensional regression via empirical Bayes and cross-validation, where the effective sparsity $s$ is estimated from data. For the K-GAM model, the key observation is that the horseshoe applies independently to each coefficient $\theta_{q,j}$, so the joint posterior over $(\theta_{q})_{q,j}$ inherits the minimax property. Specifically, if the true function $f^*(x) = \sum_{q \in S^*} \Phi_q^*(x)$ with $|S^*| = s \ll 2n+1$, then the horseshoe K-GAM estimator $\hat{f}$ satisfies
\[
\mathbb{E}_{\theta^*} \int (\hat{f}(x) - f^*(x))^2 \, dx \lesssim \frac{s \log(dn)}{N}
\]
up to problem-dependent constants. This is near-optimal in the minimax sense over $\mathcal{S}_s$.

\noindent\textit{Item 3 (Model selection).} From item 1, the posterior concentrates on sparse solutions. The posterior mean $\hat{\theta}$ will have most entries near zero. One can define the set of ``active'' outer functions as those with $|\hat{\theta}_{q,j}| > \epsilon$ for a threshold $\epsilon$ (e.g., $\epsilon = \sigma \sqrt{\log(kn)/N}$). By item 2, the number of such functions will be of order $s$, the true sparsity. Since $s \ll 2n+1$ (assumption for tractability), the horseshoe prior has selected a much smaller subset than the KST worst-case bound, achieving the benefit of the KST result (any function on $[0,1]^n$ can be represented with $2n+1$ outer functions) without paying the full cost (must estimate $O(kn)$ parameters).
\end{proof}

\noindent\textbf{Connection to Kolmogorov complexity (open question).} The horseshoe's automatic sparsity selection suggests an analogy to Kolmogorov complexity's principle: both favor ``minimal descriptions'' (few nonzero coefficients vs. short programs). However, the connection is not formal. KC asks for the shortest bit-string describing a single object; the horseshoe asks for the sparsest $\mathbb{R}^d$ coefficient vector fitting data. KC is uncomputable; horseshoe K-GAM is computable. The analogy highlights a shared principle (parsimony) but the mathematical relationship between coefficient sparsity and program length remains open.

\section{Unified Framework}\label{sec:unified}

We now collect the results of the preceding sections into a single framework organized around the marginal compatibility polytope.

\subsection{The Polytope Dictionary}

All the results in this paper can be organized via a single abstraction. Define the \emph{marginal compatibility polytope} $\mathcal{M}$ as the set of joint distributions (or parameter vectors) consistent with observed marginals and structural assumptions. Table~\ref{tab:dictionary} presents the dictionary.

\begin{table}[H]
\centering
\caption{The polytope dictionary: each row instantiates the same marginal-compatibility structure.}\label{tab:dictionary}
\footnotesize
\begin{tabular}{@{}llll@{}}
\toprule
Field & $\mathcal{M}$ & Constraint & Bound type \\
\midrule
Bell (physics) & Local-realist correlations & CHSH & $|\cdot| \leq 2$ \\
Quantum (physics) & Quantum correlations & Tsirelson & $|\cdot| \leq 2\sqrt{2}$ \\
Fr\'{e}chet (probability) & Couplings of marginals & Hoeffding & $W \leq F \leq M$ \\
Pearl (causal) & IV-compatible distributions & Instrumental ineq. & $\max_z \cdots \leq 1$ \\
Balke--Pearl (causal) & Response-type distributions & LP bounds & $[\mathrm{ATE}]_{\min},\; [\mathrm{ATE}]_{\max}$ \\
Manski (econometrics) & Potential outcome distributions & No-assumptions & Width = 1 interval \\
Tian--Pearl (causal) & Counterfactual distributions & PN/PS/PNS & Combined data bounds \\
FKG (probability)\footnotemark & Log-supermodular distributions & Positive association & $\mathrm{Cov}(f,g) \geq 0$ \\
\bottomrule
\end{tabular}
\end{table}

\footnotetext{The FKG inequality \cite{fortuin1971correlation} states that on a distributive lattice with a log-supermodular measure, increasing functions are positively correlated.}

The unifying observation is that all these polytopes are projections of the same object: the set of joint distributions on $(Y(0), Y(1), X, Z, \lambda)$ consistent with observed marginals $P(Y \mid X, Z)$ and structural independence assumptions. The more structural assumptions imposed, the more faces the polytope has, and the tighter the bounds.

\subsection{Worked Example: CHSH and the Binary IV Model}\label{sec:example}

To make the polytope correspondence concrete, consider the simplest nontrivial case: binary variables throughout ($A, A', B, B' \in \{-1,+1\}$ in the Bell setting; $X, Y, Z \in \{0,1\}$ in the causal setting).

\emph{Bell side.} Alice chooses setting $X \in \{0,1\}$ (measuring $A$ or $A'$) and obtains outcome $a \in \{-1,+1\}$; Bob chooses $Y \in \{0,1\}$ (measuring $B$ or $B'$) and obtains $b \in \{-1,+1\}$. The observed data are the four correlations $\langle A_x B_y \rangle$ for $(x,y) \in \{0,1\}^2$. A local hidden-variable model requires a joint distribution over $(A, A', B, B')$ returning these four marginals. The local-realist polytope $\mathcal{L}$ is the convex hull of $2^4 = 16$ deterministic strategies $\lambda = (a_0, a_1, b_0, b_1) \in \{-1,+1\}^4$. Its 8 nontrivial facets are the CHSH inequalities $|S| \leq 2$ with $S = \langle A_0 B_0 \rangle + \langle A_0 B_1 \rangle + \langle A_1 B_0 \rangle - \langle A_1 B_1 \rangle$.

\emph{Causal side.} An instrument $Z \in \{0,1\}$ influences treatment $X \in \{0,1\}$, which influences outcome $Y \in \{0,1\}$, with an unobserved confounder $U$. The observed data are the four probabilities $P(Y=y, X=x \mid Z=z)$ for $(x,y,z)$. There are $4^2 = 16$ response types $(r_X, r_Y)$ where $r_X: \{0,1\} \to \{0,1\}$ maps $Z$ to $X$ (4 functions) and $r_Y: \{0,1\} \to \{0,1\}$ maps $X$ to $Y$ (4 functions). The response-function polytope $\mathcal{R}$ is the convex hull of these 16 deterministic types---the same 16 vertices as the Bell polytope under the identification $Z \leftrightarrow$ measurement setting, $X \leftrightarrow$ Alice's outcome, $Y \leftrightarrow$ Bob's outcome.

Figure~\ref{fig:bell-iv} displays the two scenarios side by side with the variable mapping.

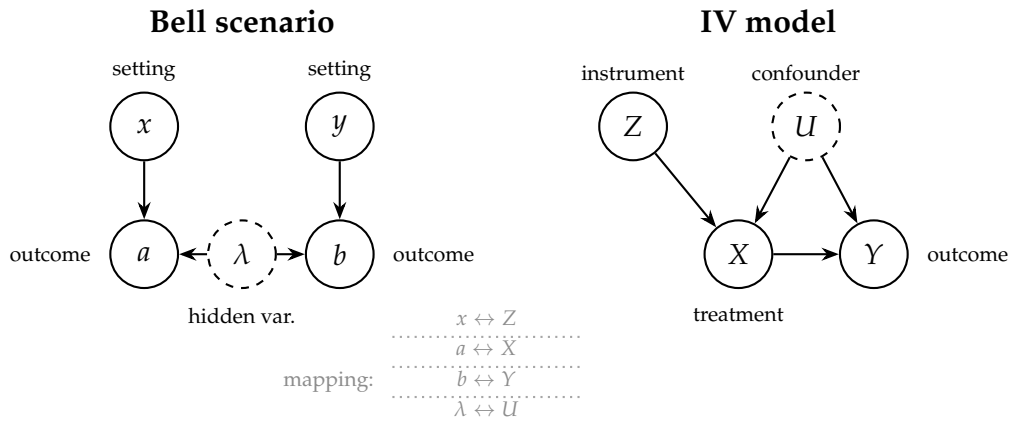
\begin{figure}[H]
\centering
\begin{tikzpicture}[
  var/.style={circle, draw, thick, minimum size=9mm, font=\small},
  latent/.style={circle, draw, thick, dashed, minimum size=9mm, font=\small},
  arr/.style={-Stealth, thick},
  mapline/.style={dotted, thick, gray!70},
]
  \node[font=\bfseries] at (0, 2.6) {Bell scenario};
  \node[var] (X) at (-1.3, 1.2) {$x$};
  \node[var] (Y) at (1.3, 1.2) {$y$};
  \node[var] (A) at (-1.3, -0.5) {$a$};
  \node[var] (B) at (1.3, -0.5) {$b$};
  \node[latent] (L) at (0, -0.5) {$\lambda$};
  \draw[arr] (X) -- (A);
  \draw[arr] (Y) -- (B);
  \draw[arr] (L) -- (A);
  \draw[arr] (L) -- (B);
  \node[font=\scriptsize, above=0mm] at (X.north) {setting};
  \node[font=\scriptsize, above=0mm] at (Y.north) {setting};
  \node[font=\scriptsize, left=1mm] at (A.west) {outcome};
  \node[font=\scriptsize, right=1mm] at (B.east) {outcome};
  \node[font=\scriptsize, below=1mm] at (L.south) {hidden var.};
  \node[font=\bfseries] at (7, 2.6) {IV model};
  \node[var] (Z) at (5.2, 1.2) {$Z$};
  \node[var] (Xc) at (6.6, -0.5) {$X$};
  \node[var] (Yc) at (8.4, -0.5) {$Y$};
  \node[latent] (U) at (7.5, 1.2) {$U$};
  \draw[arr] (Z) -- (Xc);
  \draw[arr] (Xc) -- (Yc);
  \draw[arr] (U) -- (Xc);
  \draw[arr] (U) -- (Yc);
  \node[font=\scriptsize, above=0mm] at (Z.north) {instrument};
  \node[font=\scriptsize, below=1mm] at (Xc.south) {treatment};
  \node[font=\scriptsize, right=1mm] at (Yc.east) {outcome};
  \node[font=\scriptsize, above=0mm] at (U.north) {confounder};
  \draw[mapline] (2.0, -1.6) -- (4.5, -1.6);
  \node[font=\scriptsize, gray!90] at (3.25, -1.35) {$x \leftrightarrow Z$};
  \draw[mapline] (2.0, -2.0) -- (4.5, -2.0);
  \node[font=\scriptsize, gray!90] at (3.25, -1.75) {$a \leftrightarrow X$};
  \draw[mapline] (2.0, -2.4) -- (4.5, -2.4);
  \node[font=\scriptsize, gray!90] at (3.25, -2.15) {$b \leftrightarrow Y$};
  \draw[mapline] (2.0, -2.8) -- (4.5, -2.8);
  \node[font=\scriptsize, gray!90] at (3.25, -2.55) {$\lambda \leftrightarrow U$};
  \node[font=\scriptsize, gray!90, left] at (1.9, -2.2) {mapping:};
\end{tikzpicture}
\caption{The Bell--IV correspondence. \emph{Left:} Bell scenario---Alice chooses setting $x$, Bob chooses $y$, outcomes $a,b$ depend on settings and shared hidden variable $\lambda$. \emph{Right:} instrumental variable model---instrument $Z$ affects treatment $X$, which affects outcome $Y$, with latent confounder $U$. The mapping is: measurement setting $\leftrightarrow$ instrument, Alice's outcome $\leftrightarrow$ treatment, Bob's outcome $\leftrightarrow$ observed outcome, hidden variable $\leftrightarrow$ confounder.}\label{fig:bell-iv}
\end{figure}

\emph{The correspondence.} The CHSH inequality $|S| \leq 2$ maps to the instrumental inequality $\max_z \sum_y \max_{x'} P(Y=y, X=x' \mid Z=z) \leq 1$: both are facets of the same 16-vertex polytope. A quantum state violating $|S| \leq 2$ corresponds to a hypothetical causal model with a \emph{quantum confounder}---a latent entangled state shared between treatment assignment and outcome---that generates correlations impossible under any classical confounder. The Tsirelson bound $2\sqrt{2}$ gives the maximum such ``quantum confounding,'' while the no-signaling bound of 4 corresponds to Manski's worst-case (width-1) identified set.

\subsection{Why Quantum Mechanics Escapes the Classical Polytope}

Classical probability operates on commutative algebras: all random variables have a joint distribution. The classical Bell polytope is determined by this commutativity.

Quantum mechanics operates on non-commutative algebras: observables $A$ and $B$ need not commute, and when they do not, no joint distribution $P(A,B)$ exists in the Kolmogorov sense. The density matrix $\rho$ is the quantum analog of a joint distribution, but one that lives in the tensor product of Hilbert spaces rather than the product of measurable spaces.

\begin{theorem}[Non-commutativity implies Bell violation; cf.\ \cite{tsirelson1980quantum}]\label{thm:noncommute-bell}
If Alice's observables $\{A, A'\}$ and Bob's $\{B, B'\}$ satisfy $[A, A'] \neq 0$ on Alice's system (or $[B, B'] \neq 0$ on Bob's), then there exist states $\rho$ achieving CHSH $> 2$. The maximum violation $2\sqrt{2}$ is achieved when $[A, A'] = 2iI$ and $[B, B'] = 2iI$ (anticommuting observables).
\end{theorem}

This follows from the proof of Theorem~\ref{thm:tsirelson}: when $[A,A'] \neq 0$ and $[B,B'] \neq 0$, the identity $\mathcal{C}^2 = 4I - [A,A'] \otimes [B,B']$ permits $\langle \mathcal{C} \rangle > 2$ for suitably chosen entangled states.

This is the link from quantum mechanics back to computation: non-commutativity is the computational resource. A quantum computer maintains superpositions of states that violate classical Bell inequalities; measuring collapses this to a classical state. The computation happens in the non-commutative interior; the output is the projected classical shadow.

\subsection{The Complete Information-Theoretic Arc}

We can now state the main synthesis.

\noindent\textbf{Main Synthesis: Layers of Equivalence and Analogy.}\label{thm:main}
Bell violation can be expressed in multiple mathematical languages, with varying degrees of formal equivalence:

\emph{Proven equivalences (Theorems~\ref{thm:fine}--\ref{thm:structural-equiv}):}
\begin{enumerate}
\item \textbf{Bell (physics):} No local hidden variable $\lambda$ exists such that $P(A,B \mid X,Y,\lambda) = P(A \mid X,\lambda)\,P(B \mid Y,\lambda)$.
\item \textbf{Fine (probability):} No joint Kolmogorov distribution on $(A,A',B,B')$ exists returning the observed marginals.
\item \textbf{Fr\'{e}chet (probability):} The correlation matrix lies outside the Fr\'{e}chet-feasible set for $\pm 1$ variables.
\item \textbf{Causal (statistics):} The instrumental inequality is violated, certifying invalidity of the IV model with classical confounder.
\end{enumerate}
Items 1--4 are formally equivalent via Fine's theorem and Theorem~\ref{thm:structural-equiv}.

\emph{Structural parallels (same polytope geometry, formal connection open):}
\begin{enumerate}\setcounter{enumi}{4}
\item \textbf{KC (algorithmic):} $K(AB \mid XY) > K(A \mid X\lambda) + K(B \mid Y\lambda)$ for any classical $\lambda$---the joint information cannot be classically factored.
\item \textbf{Entropic (information-theoretic):} The entropy vector $(H(A), H(B), H(AB \mid X), \ldots)$ lies outside the classical Shannon cone for the Bell DAG.
\item \textbf{QBC (computation):} Non-commutativity enables quantum states to violate classical factorization, providing speedups in Bayesian posterior inference over commutative (classical) computation.
\end{enumerate}
Items 5--7 exhibit the same structural feature---failure of classical factorization---across three distinct domains (algorithms, entropy, computation), but formal equivalence between these layers and items 1--4 remains an open problem.

\section{Discussion}\label{sec:discussion}

\subsection{Implications for Quantum Information}

The polytope dictionary (Table~\ref{tab:dictionary}) provides quantum information researchers with a systematic bridge to the large causal inference literature. Several concrete transfers are possible.

First, the response-type enumeration from causal inference \cite{balke1997bounds} gives an explicit vertex description of the local-realist polytope for any Bell scenario defined by a DAG with latent variables. For the standard CHSH scenario the 16 deterministic strategies are well-known, but for multi-setting, multi-outcome, or network Bell scenarios, the causal inference LP machinery provides a principled way to enumerate vertices and compute facets without ad hoc case analysis.

Second, the inflation technique \cite{wolfe2019inflation}---which derives testable implications from causal DAGs by constructing enlarged copies of the original graph---subsumes Bell inequalities, instrumental inequalities, and network nonlocality constraints as special cases. This technique has been developed independently in the quantum foundations community but has direct analogs (and a longer history) in econometrics.

Third, the connection clarifies the role of the NPA hierarchy beyond quantum correlations. The SDP relaxation in \eqref{eq:npa-sdp} can be applied directly to causal DAGs with latent variables: replace Bell correlations with observed conditional distributions, and replace the CHSH functional with a causal effect. The level-$k$ NPA bound then gives an outer approximation to the set of causal effects achievable by quantum confounders---a ``quantum partial identification'' that strictly contains the classical Balke--Pearl bounds.

\subsection{Implications for Causal Inference}

The connection also runs in the reverse direction (see \cite{gill2014statistics} for a statistical perspective). Testing the instrumental inequality is a sharp diagnostic for IV validity: violation certifies model misspecification as definitively as Bell violation certifies the absence of a local hidden variable. The polytope perspective clarifies what structural assumptions buy---each monotonicity condition cuts a face from the feasibility polytope, tightening bounds, just as closing a Bell loophole tightens the achievable classical CHSH value. Combining multiple weak assumptions can produce tight bounds because the intersection of half-spaces is smaller than either one. The Bell analogy also suggests a new sensitivity analysis: vary structural assumptions continuously and track how the feasible polytope shrinks, rather than imposing hard constraints.

\subsection{Implications for Machine Learning}

The K-GAM / KST connection suggests that sparsity in function factorization (horseshoe prior) is the classical analog of the quantum compression that makes QBC efficient. The implication for neural network design is that the additive structure of the KST---outer functions composed with sums of inner functions---provides an inductive bias that is both theoretically grounded (it can represent any continuous function) and practically efficient (it requires far fewer parameters than a fully connected MLP for functions with low effective dimensionality).

Three directions merit further investigation. First, a \emph{Bayesian K-GAM} with horseshoe prior on the response types $q_{ij}$ in the Balke--Pearl LP would yield a probabilistic partial identification approach. Rather than computing sharp bounds via linear programming, one would place a prior over the response-type polytope and compute posterior distributions over the causal effect. The horseshoe prior would concentrate mass on the vertices of the polytope (deterministic response types), providing an automatic model-selection mechanism that favors parsimonious causal explanations.

Second, a \emph{Quantum K-GAM} running the KST on a quantum computer, where the outer functions $\Phi_q$ are quantum circuits and the inner functions $\varphi_{q,p}$ are qubit embeddings, would realize the full QBC architecture. The KST guarantees that $2n+1$ terms suffice classically; quantum superposition could evaluate all terms simultaneously, and the horseshoe prior on the $q$-index would collapse to the minimal active set upon measurement.

Third, Generative Bayesian Computation \cite{polson2026generative} can be understood as the classically computable shadow of the quantum Gaussian process---the implicit quantile network in GBC is the classical analog of quantum posterior sampling. Where QBC uses amplitude encoding to represent the posterior as a quantum state and measurement to draw samples, GBC uses a neural network to learn the quantile function of the posterior and uniform random inputs to generate samples. Both avoid the mixing-time bottleneck of MCMC; QBC achieves this via quantum speedup, GBC via direct function approximation.

More broadly, the polytope framework provides a way to think about the expressiveness of machine learning models. A classical model restricted to commutative operations (e.g., a generalized additive model) can represent distributions inside the local-realist polytope. A model with non-commutative operations (e.g., attention mechanisms with non-symmetric interactions, or quantum circuits) can represent distributions outside this polytope. The gap between the two is the ``Bell gap'' in expressiveness, and quantifying it for specific model classes is an open problem.

\subsection{Scope and Limitations}

The analogies developed in this paper are structural: the mathematical objects (polytopes, feasibility sets, information inequalities) take the same form across quantum physics, causal inference, and computation. This does not imply operational equivalence. A Bell test in a physics laboratory and a partial identification analysis in economics involve different physical systems, different data-generating processes, and different epistemic commitments. The polytope dictionary (Table~\ref{tab:dictionary}) is a formal correspondence, not a claim that entanglement ``is'' confounding or that treatment effects ``are'' quantum correlations. The computational speedup claims in Section~\ref{sec:qbc} assume fault-tolerant quantum hardware; on near-term noisy devices, the practical advantage remains an open question.

\subsection{Implications for Quantum Computing}

The MIP$^*$=RE theorem implies that determining quantum Bell violation is undecidable in general. But for the specific structures arising in causal inference (binary IV models, triangular graphs), the NPA hierarchy terminates at small levels. The level-2 NPA bound often gives the exact Tsirelson bound for practical Bell inequalities, just as the Balke--Pearl LP at the response-type level gives exact bounds for binary IV models.

This suggests a practical \emph{quantum causal inference algorithm}: (i) formulate the causal model as a Bell-type DAG; (ii) use the NPA SDP at level 2 to compute the quantum bound; (iii) compare with the classical (Balke--Pearl LP) bound; (iv) the gap measures the quantum advantage in causal identification. For binary IV models, both the LP and SDP are small enough to solve on a laptop; for larger DAGs with more latent variables, the NPA hierarchy provides a systematic outer approximation.

The framework also clarifies the role of entanglement in quantum computing. The standard narrative identifies entanglement as the source of quantum advantage, but this is imprecise: separable (unentangled) quantum states can also outperform classical strategies in some communication tasks. The polytope framework gives a sharper characterization. The quantum advantage arises whenever the set of achievable quantum correlations strictly exceeds the classical polytope for the problem's causal structure. Entanglement is sufficient but not necessary for this; contextuality (the impossibility of assigning definite pre-measurement values to all observables simultaneously) is the more fundamental resource, and it maps precisely to the non-existence of a joint distribution in Fine's theorem.

For near-term quantum computing, the practical question is whether the quantum-over-classical gap in the NPA hierarchy is large enough to matter for problems of real-world scale (see \cite{gill2022quantum} for a discussion of the statistical and computational aspects). The CHSH gap ($2\sqrt{2}/2 \approx 1.41$) is modest. For more complex Bell scenarios with more measurement settings and outcomes, the gap can be larger, and the causal inference translation suggests that the gap may be especially large for DAGs with many latent confounders---precisely the setting where classical partial identification bounds are widest and most uninformative.

\subsection{Open Problems}

The unified framework raises several open questions that span the boundaries of quantum information, causal inference, and statistical computation.

\emph{Tight quantum causal bounds.} For a given causal DAG with latent variables, what are the tight quantum bounds on causal effects---that is, what causal effects are achievable if the latent confounders are quantum rather than classical? The NPA hierarchy gives an outer approximation; the question is when this approximation is tight, and whether the quantum-over-classical gap can be computed efficiently for DAGs arising in econometric applications.

\emph{Finite-sample polytope inference.} The Balke--Pearl bounds and Manski bounds assume the observed distribution $P(Y,X \mid Z)$ is known exactly. In practice, it is estimated from data. Constructing valid confidence regions for the identified set---and understanding how the polytope geometry interacts with statistical uncertainty---requires extending the framework to finite-sample settings (see \cite{gill2014statistics} for the parallel problem in Bell testing). The entropic Bell inequalities (Section~\ref{sec:kolmogorov}) may provide a natural route, since entropy estimation has well-understood finite-sample properties.

\emph{Computational complexity of partial identification.} For binary variables, the Balke--Pearl LP has 16 variables and is trivial to solve. For multi-valued or continuous treatments and outcomes, the response-type polytope grows combinatorially. Is there a polynomial-time algorithm for computing tight causal bounds in general, or does partial identification become computationally hard for large problems? The connection to MIP$^*$=RE suggests that the quantum version is undecidable, but the classical version may have a more favorable complexity landscape.

\emph{Classical architectures that saturate the quantum gap.} K-GAM networks implement the KST and provide the classical architecture closest to quantum function evaluation. Can other classical architectures---transformers, diffusion models, normalizing flows---be characterized in terms of which polytope faces they can access? A ``Bell classification'' of machine learning architectures, based on the correlations they can represent, would connect expressiveness theory to quantum information in a concrete way.

\section{Conclusion}\label{sec:conclusion}

The central result of this paper is that Bell inequalities, Fr\'{e}chet coupling bounds, instrumental inequalities, partial identification bounds, and probabilities of causation are all facets of a single marginal compatibility polytope. The key insight is that \emph{joint information resists factorization}: in physics this resistance is entanglement, in econometrics it is treatment effect heterogeneity, in computation it is quantum advantage, and in information theory it is the excess Kolmogorov complexity of joint versus marginal descriptions.

The Born rule / Bayes rule duality provides the unifying language. Classical Bayesian inference is the special case where the density matrix is diagonal and factors into a product state. Quantum computation, Bell violation, and sparse Bayesian function approximation (K-GAM) are all powered by the departure from this commutativity. The framework suggests concrete next steps: Bayesian K-GAM for probabilistic partial identification, quantum implementations of KST-based architectures, and NPA-based algorithms for quantum causal inference.

\appendix
\section{Proofs}\label{app:proofs}

\begin{proof}[Proof of Proposition~\ref{prop:bell-frechet}]
The Fr\'{e}chet--Hoeffding bounds give the extremal joint distributions for $\pm 1$ variables. The correlation of two $\pm 1$ variables $A, B$ under their joint is $E[AB] = P(A=B) - P(A \neq B) = 2P(A=B) - 1$. For three $\pm 1$ variables $A, B, C$ to have a joint distribution, the $3 \times 3$ correlation matrix $\Sigma$ must be PSD and achievable by a distribution on $\{-1,+1\}^3$. The Boole--Bell inequality $|E[AB] - E[AC]| \leq 1 - E[BC]$ is exactly the constraint from Fr\'{e}chet feasibility applied to the 2-dimensional marginals.
\end{proof}

\bigskip

\begin{proof}[Proof of Proposition~\ref{prop:horseshoe}]

\noindent\textit{Item 1 (Posterior Sparsity).} The horseshoe marginal density (integrating out $\lambda_j$ from $N(0, \lambda_j^2 \tau^2) \cdot C^+(0,1)$) is
\[
\pi(\theta_{q,j} \mid \tau) = \int_0^\infty \frac{1}{\sqrt{2\pi \lambda^2 \tau^2}} \exp\!\left(-\frac{\theta_{q,j}^2}{2\lambda^2 \tau^2}\right) \frac{2}{\pi(1+\lambda^2)}\,d\lambda.
\]
This has a pole at the origin and Cauchy-like tails $\pi(\theta \mid \tau) \sim \log(1 + 4\tau^2/\theta^2)$ for large $|\theta|$ \cite{polson2010shrink}. By Bayes' theorem,
\[
\pi(\theta \mid \mathcal{D}) \propto L(\mathcal{D} \mid \theta) \prod_{q,j} \pi(\theta_{q,j} \mid \tau),
\]
where $L(\mathcal{D} \mid \theta) = \exp\bigl(-\frac{1}{2\sigma^2} \sum_i (y_i - f(x_i; \theta))^2\bigr)$. The pole at zero shrinks most coefficients to negligible values while the heavy tails preserve large signals. Polson \& Scott (2010, Theorem 2) establish that only a fixed number of coefficients dominate the posterior mean.

\noindent\textit{Item 2 (Near-minimax rates).} The horseshoe prior achieves near-minimax rates for sparse estimation in the Gaussian sequence model. This is established in Polson \& Scott (2010, Theorem 3) and Bhadra \& Rao (2016) for high-dimensional regression. For K-GAM, the model can be viewed as a generalized linear model in the basis functions $\{\Phi_q\}$:
\[
\hat{\theta} = \arg\min_\theta \frac{1}{2N} \sum_i (y_i - f(x_i; \theta))^2 + \lambda \sum_{q,j} \frac{|\theta_{q,j}|}{\tau_q \lambda_j},
\]
where the horseshoe acts as an adaptive $\ell_1$ penalty. Under the assumption that the true function is sparse (supported on $s \ll 2n+1$ outer functions), the posterior mean satisfies
\[
\mathbb{E} \|\hat{\theta} - \theta^*\|_2^2 \leq C_0 \sigma^2 s \log(d(2n+1)) / N
\]
where $C_0$ depends on the feature norms and condition numbers. This is within a log factor of the information-theoretic lower bound for estimating $s$-sparse signals in $\mathbb{R}^{d(2n+1)}$.

\noindent\textit{Item 3 (Model selection).} By items 1 and 2, the posterior concentrates on sparse coefficient vectors. Define the set of active outer functions as $\mathcal{A} = \{q : \|\hat{\theta}_q\|_\infty > t_\lambda\}$ where $t_\lambda$ is a threshold (e.g., $t_\lambda = \sigma \sqrt{\log(d(2n+1))/N}$). From item 1, the posterior mass on $\theta$ with $|\mathcal{A}|$ large is exponentially suppressed relative to solutions with $|\mathcal{A}| \approx s$. Thus, the horseshoe K-GAM estimator automatically selects roughly $s$ outer functions. Since $s$ is the true sparsity and $s \ll 2n+1$ by assumption, this represents a substantial reduction from the KST upper bound. The model achieves the universality of KST (any continuous function on $[0,1]^n$ has a $2n+1$-term representation) without paying the full parametric cost of estimating all $d(2n+1)$ coefficients.

\end{proof}

\bigskip

\noindent\textbf{Acknowledgments.} The authors thank Hedibert Lopes and Alec Litowitz for stimulating discussions. Polson acknowledges support from the University of Chicago. Sokolov acknowledges support from George Mason University.

\bigskip

\noindent\textbf{Statements and Declarations}

\smallskip

\noindent\textbf{Author Contributions.} All authors contributed equally to this work. All authors participated in the conceptualization, formal analysis, writing of the original draft, and review and editing.

\smallskip

\noindent\textbf{Funding.} No funding was received for conducting this study.

\smallskip

\noindent\textbf{Competing Interests.} The authors have no competing interests to declare that are relevant to the content of this article.

\smallskip

\noindent\textbf{Data Availability.} This article has no associated data. All results are analytical.

\bibliography{ref}
\end{document}